\documentclass[a4paper, 12pt]{amsart}
\usepackage[utf8]{inputenc}
\usepackage[english]{babel}
\usepackage{amsmath, amsthm, amssymb}
\usepackage{graphicx}
\usepackage[dvipsnames]{xcolor}
\usepackage{caption}
\usepackage{subcaption}
\usepackage{mathtools}
\usepackage{enumerate}
\usepackage{tikz}
\usepackage{tikz-cd}
\usepackage{enumitem}
\usepackage{hyperref}
\hypersetup{ 
	colorlinks = true,
	allcolors = red,
}
\usepackage[capitalize,noabbrev]{cleveref}
\usepackage{ytableau}
\usepackage{faktor}

\usepackage[a4paper]{geometry}
\usepackage{titlesec}
\titleformat{\section}
{\normalfont\Large\bfseries}{\thesection}{1em}{}[{\titlerule[0.8pt]}]

\newtheorem{theorem}{Theorem}[section]
\newtheorem{proposition}[theorem]{Proposition}
\newtheorem{lemma}[theorem]{Lemma}
\newtheorem{corollary}[theorem]{Corollary}

\theoremstyle{definition}
\newtheorem{definition}[theorem]{Definition}

\newtheorem{remark}[theorem]{Remark}

\newtheorem{example}[theorem]{Example}

\newtheorem{question}[theorem]{Question}

\DeclareMathOperator{\dist}{dist}

\DeclareMathOperator{\mindist}{dist}
\DeclareMathOperator{\val}{val}

\DeclareMathOperator{\GL}{GL}
\DeclareMathOperator{\SL}{SL}

\DeclareMathOperator{\End}{\mathrm{End}}

\DeclareMathOperator{\B}{B}
\DeclareMathOperator{\rev}{rev}
\DeclareMathOperator{\Aut}{Aut}
\DeclareMathOperator{\Gr}{Gr}
\DeclareMathOperator{\card}{card}
\DeclareMathOperator{\Sym}{Sym}
\DeclareMathOperator{\Lcal}{\mathcal{L}}
\DeclareMathOperator{\Ccal}{\mathcal{C}}

\newcommand{\R}{\mathbb{R}}
\newcommand{\Q}{\mathbb{Q}}

\newcommand{\Z}{\mathbb{Z}}

\newcommand{\F}{\mathbb{F}}
\newcommand{\Ocal}{\mathcal{O}}

\newcommand{\Bcal}{\mathcal{B}}

\newcommand{\minplus}{\,\underline{\oplus}\,}
\newcommand{\maxplus}{\,\overline{\oplus}\,}

\newcommand{\bigminplus}{\,\underline{\bigoplus}\,}

\newcommand{\msc}[1]{\href{https://zbmath.org/classification/?q=#1}{#1}}
\newcommand{\gen}[1]{\langle #1\rangle}

\numberwithin{equation}{section}

\makeatletter
\renewcommand\subsection{\@startsection{subsection}{2}%
	\z@{.5\linespacing\@plus.7\linespacing}{-.5em}%
	{\normalfont \bfseries}}
\makeatother

\makeatletter
\@namedef{subjclassname@2020}{%
	\textup{2020} Mathematics Subject Classification}
\makeatother

\begin{document}
	
	\title{Submodule codes as spherical codes in buildings}
	
	\author[M.~Stanojkovski]{Mima Stanojkovski}
	\address[Mima Stanojkovski]{Universit\`a di Trento, Dipartimento di Matematica}
	\email{mima.stanojkovski@unitn.it}
	
	
	
	\date{\today}

	\begin{abstract}
We give a generalization of subspace codes by means of codes of modules over finite commutative chain rings. We define a new class of Sperner codes and use results from extremal combinatorics to prove the optimality of such codes in different cases. Moreover, we explain the connection with Bruhat--Tits buildings and show how our codes are the buildings' analogue of spherical codes in the Euclidean sense.
\end{abstract}	
	
	\maketitle
	

\section{Introduction}

\noindent
The codes studied in this paper can be viewed as a bridge of generalization between two worlds, that of subspace codes and that of spherical codes. More specifically our codes consist of \emph{equivalence classes} of modules over finite commutative \emph{chain rings}, which can be interpreted at the same time as subsets of spheres in \emph{Bruhat-Tits buildings}. In this introduction we will take a first glance at these connections and present the main questions that will be addressed in this document. 

\subsection{Spherical codes in the Euclidean setting}

Spherical codes in $\R^d$, equipped with the usual distance, are finite subsets of the unit sphere 
\[
\B_1=\{x=(x_1,\ldots,x_d)\in \R^{d} \mid x_1^2+\ldots+x_d^2=1\}.
\]
In this context, spherical codes can be constructed from sphere packings \cite[Section~1.2.4]{CS93} and find numerous applications in the field of telecommunication. In view of the applications, it is desirable to produce sizable codes of large internal distance and small length. Optimal codes are thus codes with the ``best possible'' coexistence constraints on the last requirements. More precisely, it is greatly interesting to determine which spherical codes present the most favourable relationship between their \emph{length}, \emph{minimum distance}, and \emph{cardinality}. Already for the small length value $d=3$, however, the last problem turns out to be very hard and not all \emph{optimal} codes are classified; cf.\ \cite[Section~3.3]{EriZin/01}. For a broad overview of spherical codes in this setting we refer the interested reader to \cite{EriZin/01}.

\subsection{Chain rings in coding theory}

Let $R$ be a commutative ring, which we assume to be unital. The ring $R$ is said to be a \emph{chain ring} if all of its ideals form a chain, i.e.\ if $I$ and $J$ are ideals of $R$, then $I\subseteq J$ or $J\subseteq I$. In this paper, only the case of commutative chain rings will be considered, though their definition extends also to the non-commutative case; cf.\ \cite[Section 2]{HonLan00}. Examples of finite commutative chain rings include 
\begin{enumerate}[label=$(\arabic*)$]
    \item $\Z/p^r\Z$, where $p$ is a prime number and $r$ a positive integer, and
    \item $(\Z/p^m\Z)[x]/(f(x))$, where $p$ is a prime number, $m$ a positive integer, and $f$ a monic polynomial that is irreducible modulo $p$.
\end{enumerate}
For more on the classification of finite commutative chain rings we refer to \cite{AA/22,ClaLi73,Hou01}.
In the present paper, we are mostly interested in viewing $R$ as a quotient of a discrete valuation ring $\Ocal_K$ by a power $\mathfrak{m}_K^r$ of its unique maximal ideal $\mathfrak{m}_K$, e.g.\ $\Ocal_K$ equals the $p$-adic integers $\Z_p$ or the ring $\F_q[[t]]$ of formal power series with nonnegative integer exponents and coefficients in the field $\F_q$. As can be found for instance in \cite[Section 2]{HonLan00}, finite chain rings are local and their unique maximal ideal $\mathfrak{m}$ is principal. Moreover, if $\pi$ generates $\mathfrak{m}$, then every ideal of $R$ is generated by a nonnegative power of $\pi$. Since $R$ is finite, there exists a minimal positive integer $r$, called the nilpotency class of $\mathfrak{m}$, with the property that $\mathfrak{m}^r=0$, equivalently that $\pi^r=0$. In addition, an elementary divisor type theorem holds for finitely generated modules over chain rings. There are several applications of finite chain rings in coding theory including linear codes \cite{Bla75,HonLan00,LiuLiu15} and cyclic codes \cite{CalSlo95,DiLo04,Gre97,NoSu00}, though to this author's best knowledge the consideration of codes consisting of modules over finite chain rings does not appear anywhere in the literature.

\subsection{Spherical codes of modules}
Let $R$ be a finite commutative chain ring and let $r\geq 1$ be such that the unique maximal ideal $\mathfrak{m}$ of $R$ satisfies $\mathfrak{m}^{r-1}\neq0$ and $\mathfrak{m}^r=0$. Let $\pi\in R$ satisfy $\mathfrak{m}=R\pi$ and let $V_r$ be a free $R$-module of rank $d\geq 2$, that is $V_r\cong R^d$. Write $\mathcal{L}(V_r)$ for the set of all $R$-submodules of $V_r$ and $\partial\mathcal{L}(V_r)$ for the \emph{boundary} of $\mathcal{L}(V_r)$:
\[
\partial\mathcal{L}(V_r)=\{U\in \mathcal{L}(V_r) \mid \pi^{r-1}V_r\not\subseteq U\not\subseteq \pi V_r\}.
\]
Defining the map $\dist:\partial\mathcal{L}(V_r)\times \partial\mathcal{L}(V_r)\rightarrow\Z$ by 
\[
(U_1,U_2)\mapsto \dist(U_1,U_2)=\min\{m\in\Z_{\geq 0} \mid \pi^{m}U_1\subseteq U_2\}+\min\{n\in\Z_{\geq 0} \mid \pi^{n}U_2\subseteq U_1\}
\]
gives $\partial\mathcal{L}(V_r)$ the additional structure of a metric space. The last distance can be extended to the whole of $\mathcal{L}(V_r)$ modulo \emph{homothety}; cf.\ \cref{sec:module-distance}. Moreover, for $r=1$, 
one can see that $\dist$ does not coincide with the subspace metric or the injection metric on $\mathcal{L}(V_1)$; cf.\ \cite[Section~1]{KSK09}.  
A \emph{spherical code} in $V_r$ is then a subset $\Ccal$ of $\partial\mathcal{L}(V_r)$ of cardinality at least $2$ and its \emph{minimum distance} is 
\[
\dist(\Ccal)=\min\{\dist(U_1,U_2)\mid U_1,U_2\in\Ccal,\, U_1\neq U_2\}.
\]
Spherical codes in $V_r$ are natural generalizations of subspace codes, though the attribute ``spherical'' comes from interpreting $\partial\mathcal{L}(V_r)$ as a sphere of modules, cf.\ \cref{prop:ball-sphere-modules}. 
In this manuscript, we address and give answers to the following question:
\begin{center}
{\em
For a given integer $\psi$, what are the largest spherical codes $\mathcal{C}$ in $V_r$ 

with the property that $\mindist(\Ccal)\geq\psi$? 
}
\end{center}
The largest codes associated to a given minimum distance are called \emph{optimal}. 
If $\psi=1$, then there is a unique optimal code of minimum distance $1$, namely $\partial\mathcal{L}(V_r)$: we compute its cardinality in \cref{sec:counting}. In general, good candidates for optimal codes  are the \emph{Sperner codes} that we define in \cref{sec:Sperner} using Grassmannians of $R$-modules. Such codes are defined starting from the parameters $(d, R,\alpha)$ where $\psi=2\alpha$ is taken to be even.
In \cref{thm:sperner-codes}, we compute the cardinality and minimum distance of a Sperner code with parameters $(d,R,\alpha)$, yielding general bounds on the maximal size of codes of minimum distance $2\alpha$; cf.\ \cref{cor:general-bound}. In \cref{sec:extremal}, we use results from extremal combinatorics to prove that Sperner codes are optimal when $\alpha=r$ or $d=2$; cf.\ \cref{th:dist=2r} and \cref{th:d=2}. We move on to the construction, in \cref{sec:permutation}, of optimal codes in a subfamily of $\partial\mathcal{L}(V_r)$ indexed by tuples of positive integers.
More concretely, let $\partial\mathcal{L}_{\bf e}(V_r)$ denote the collection of boundary $R$-submodules of $V_r$ that can be generated \emph{compatibly with} a basis ${\bf e}=(e_1,\ldots,e_d)$ of $V_r$ over $R$, i.e.\ modules of the form
\[
U=R\pi^{\delta_1}e_1\oplus\ldots \oplus R\pi^{\delta_d}e_d, \textup{ where } 0\leq \delta_i\leq r,\ \{0,r\}\subseteq \{\delta_1,\ldots,\delta_d\}.
\]
Generalizing \cite[Chapter~4]{EriZin/01}, a \emph{permutation code} is a spherical code in $V_r$ that is contained in $\partial\mathcal{L}_{\bf e}(V_r)$ and whose elements form one orbit under the natural action of the symmetric group $\Sym(d)$ on $\partial\mathcal{L}_{\bf e}(V_r)$. In \cref{th:permutation} we give bounds on minimum distance and cardinality of a permutation code in terms of its defining parameters. 

\subsection{The connection to Bruhat-Tits buildings}

Write $R=\Ocal_K/\mathfrak{m}_K^r$ and consider the natural projection $\Ocal_K^d\rightarrow R^d\cong V_r$. Via the last map we identify every submodule of $V_r$ with the unique maximal free $\Ocal_K$-submodule of $\Ocal_K^d$ mapping to it. Such a module is called a \emph{lattice} in $K^d$. The collection of lattices in $K^d$, considered up to \emph{homothety}, forms the collection of $0$-simplices of the \emph{Bruhat-Tits Building} $\mathcal{B}_d(K)$ of $\SL_d(K)$. In this infinite simplicial complex, $s$-simplices are given by chains $L_1\supset L_2\supset\ldots\supset L_s\supset \pi L_1$ of lattices and maximal simplices all have size $d$. Transporting $\dist$ from the module setting to the buildings context (see \cref{sec:dist-balls-buildings}) via the above projection, one can then interpret $\partial\mathcal{L}(V_r)$ as 
a sphere $\partial\B_r$ in $\mathcal{B}_d(K)$, cf.\ \cref{th:isometry}, and thus spherical codes in $V_r$ as spherical codes in $\B_r$.
To the best of our knowledge, this is the very first instance in which spherical codes in \emph{affine} buildings are studied, adding  yet another item to the already long list of applications of buildings; cf.\  \cref{sec:spherical-in-BT}. In our closing \cref{sec:counting} we give formulas and asymptotics for the number of elements in a ball of radius $r$ in $\mathcal{B}_d(K)$; cf.\ \cref{th:asymptotics}. As a consequence, we derive densities of spherical codes in analogy to the ones found in 
 \cite{BHKW22} for linear codes over finite chain rings. 

\subsection{A note on the underlying geometry and combinatorics}

Contrarily to what happens in the Euclidean context, a sphere in $\mathcal{B}_d(K)$ is not a homogeneous space, but is rather to be thought of as the collection of boundary points of a lattice polytope and Sperner codes arise as strategically chosen subsets of the polytope's vertices. 
As we deal with a discrete set, it is interesting and important to understand how the number of elements of $\partial\B_r$ depends on the size $q$ of the residue field of $K$. This count and its asymptotic behaviour has been included in \cref{sec:counting} as it seemed not to be explicitly available in the literature already. The count is much easier and independent of $q$ when one restricts to the analogue $\partial\B_r\cap\mathcal{A}$ of $\mathcal{L}_{\bf e}(V_r)$ in the building. Indeed, in such case we are considering a slice of $\partial\B_r$ by an affine $d$-dimensional space resulting in a polytope that is both convex in the usual and in the tropical sense; cf.\ \cref{sec:permutation}.

\subsection{Notation}
Throughout the paper, let $d\geq 2$ and $r\geq 1$ denote two integers. 
Let $R$ be a finite commutative chain ring with maximal ideal $\mathfrak{m}$ generated by $\pi$ and such that $\pi^r=0$, but $\pi^{r-1}\neq 0$.  Write $q=|R/\mathfrak{m}|$ for the cardinality of the residue field of $R$. Let $V_r$ denote a free $R$-module of rank $d$ and fix ${\bf e}=(e_1,\ldots,e_d)$ to be a basis of $V_r$ over $R$. If $r=1$, then $R$ is a field and we simply write $V=V_1$. Let ${\bf 1}$ denote the vector $(1,\ldots,1)\in\Z^d$, let $\Sym(d)$ denote the symmetric group on $d$ letters, and let $J_d$ denote the integral $(d\times d)$-matrix with $0$'s
 on the diagonal and $1$'s elsewhere. 
 Set, additionally
	\begin{align*}
	    \mathcal{E}_r^{(d)} 
	    &=\{\varepsilon=(\varepsilon_1,\ldots,\varepsilon_{d-1},\varepsilon_d=0)\in\Z^d \mid r\geq \varepsilon_1\geq\ldots\geq \varepsilon_{d-1}\geq 0\}, \\
	    \partial \mathcal{E}_r^{(d)}
	    & =\{\varepsilon=(\varepsilon_1=r,\varepsilon_2,\ldots,\varepsilon_{d-1},\varepsilon_d=0)\in\Z^d \mid r\geq \varepsilon_2\geq\ldots\geq \varepsilon_{d-1}\geq 0\}.
	\end{align*}
 In conclusion, for an indeterminate $X$, integers $a\geq b\geq 0$, and  $I=\{i_1,\ldots,i_{\ell}\}\subseteq \Z_{\geq 0}$, put
	\begin{align*}
	    \binom{a}{b}_X  & = \prod_{i=0}^{b-1}\frac{1-X^{a-i}}{1-X^{b-i}} \ \ \textup{ and }\ \
	    \binom{d}{I}_X  = \binom{d}{i_{\ell}}_X\binom{d}{i_{\ell-1}}_X\cdots \binom{d}{i_{1}}_X.
	\end{align*}
 
\section{The module distance}\label{sec:module-distance}

\noindent
In this section we define an equivalence relation on the set $\mathcal{L}(V_r)$ of all $R$-submodules of $V_r$ and a distance on the collection of its equivalence classes. 

\begin{definition}
Let $U$ be an element of $\mathcal{L}(V_r)$. Then $m_U\in\{0,\ldots,r\}$ is defined as
\[
m_U=
\max\{0\leq m\leq r\mid U\subseteq \pi^mV_r\}.
\]
Moreover, $\tilde{U}\in\mathcal{L}(V_r)$ is defined to be the kernel of the map 
\[V_r\longrightarrow V_r/U, \quad x\longmapsto \pi^{m_U}x+U.\]
\end{definition}

\noindent 
Note that $\tilde{U}$ is the unique maximal $R$-submodule of $V_r$ with the property that $\pi^{m_U}\tilde{U}=U$. 
In particular, we have that $U\subseteq \tilde{U}$ and, moreover, $m_U=0$ if and only if $\tilde{U}=U$. As a consequence, we have that 
\begin{align}
    \partial\mathcal{L}(V_r) & = \{U\in \mathcal{L}(V_r) \mid \pi^{r-1}V_r\not\subseteq U\not\subseteq \pi V_r\}\label{eq:def-boundary} \\
    & \subseteq  \{ U\in \mathcal{L}(V_r) \mid \tilde{U}=U\}.\label{eq:inclusion-tilde}
\end{align}

\begin{definition}
Modules $U$ and $U'$ in $\mathcal{L}(V_r)$ are \emph{homothetic} whenever $\tilde{U}=\tilde{U'}$.
\end{definition}

\noindent
Homothety defines an equivalence relation $\sim$ on $\mathcal{L}(V_r)$ and we write
\[
\mathcal{L}^0(V_r)=\ \faktor{\mathcal{L}(V_r)}{\sim}\ = \{\,[U]=\{U'\in \mathcal{L}(V_r) \mid \tilde{U}=\tilde{U'}\} \mid U\in\mathcal{L}(V_r)\}
\]
for the collection of homothety classes of elements of $\mathcal{L}(V_r)$. Note that $[V_r]=[0]$ has cardinality $r+1$ and the cardinality of each $[U]\in\mathcal{L}^0(V_r)\setminus\{[0]\}$ is at most $r$. Moreover, it is not difficult to see that  $\partial\mathcal{L}(V_r)$ can be identified with the collection of equivalence classes in $\mathcal{L}^0(V_r)$ with exactly one element. With a slight abuse of notation, we thus write
\begin{equation}\label{eq:boundary-class}
\partial\mathcal{L}(V_r)=\{[U]\in\mathcal{L}^0(V_r) \mid [U]=\{U\}\}.
\end{equation}

\noindent
We define a metric on $\mathcal{L}^0(V_r)$, which does not generalize the subspace or the injection metric; cf.\ \cite[Section~1]{KSK09}.

\begin{definition}\label{def:dist-mod}
	Let $[U],[U_1],[U_2] \in {\mathcal L}^0(V_r)$ denote
	homothety classes of modules.
	Define 
	\begin{align*}
	n_{12} & =\min\{m\in\Z_{\geq 0} \mid \pi^{m}\tilde{U_1}\subseteq \tilde{U_2}\} \textup{ and }\\
	n_{21}&=\min\{n\in\Z_{\geq 0} \mid \pi^{n}\tilde{U_2}\subseteq \tilde{U_1}\}.
	\end{align*}
	Then the \emph{distance} between $[U_1]$ and $[U_2]$ is
$$\dist([U_1],[U_2])= n_{12}+n_{21} .$$
	For a subset ${\mathcal M}  \subseteq {\mathcal L}^0(V_r)$, put $\dist([U],{\mathcal M} ) = \min \{\dist ([U],[U']) \mid [U']\in\mathcal{M} \}$.
	\end{definition}


\noindent
The next result gives that $\mathcal{L}^0(V_r)$ equipped with $\dist$ is a metric space.

\begin{lemma}\label{lem:is-distance}
The map $\dist:\mathcal{L}^0(V_r)\times\mathcal{L}^0(V_r)\rightarrow\Z$ is a distance.
\end{lemma}

\begin{proof}
We only show that the triangle inequality holds, as the other defining properties are clear. For this, let $[U_1],[U_2],[U_3]\in\mathcal{L}^0(V_r)$ and, for $i,j=1,2,3$, let $n_{ij}$ be as in \cref{def:dist-mod}. 
It follows from their definitions that 
\[
\pi^{n_{13}}\tilde{U_1}\subseteq \tilde{U_3}, \quad \pi^{n_{32}}\tilde{U_3}\subseteq \tilde{U_2}, \quad 
\pi^{n_{23}}\tilde{U_2}\subseteq \tilde{U_3}, \quad 
\pi^{n_{31}}\tilde{U_3}\subseteq \tilde{U_1},
\]
and so the minimalities of $n_{12}$ and $n_{21}$ yield 
\[
n_{12}\leq n_{13}+n_{32} \textup{ and } n_{21}\leq n_{23}+n_{31}.
\]
It follows from \cref{def:dist-mod} that $\dist([U_1],[U_2])\leq \dist([U_1],[U_3])+\dist([U_2],[U_3])$.
\end{proof}

\noindent
We remark that every element in $\mathcal{L}^0(V_r)$ has distance at most $r$ from $[V_r]$, equivalently the set $\mathcal{L}^0(V_r)$ can be interpreted as the \emph{ball of radius $r$ around $[V_r]$}:
\begin{equation}
    \mathcal{L}^0(V_r)=\B_r([V_r])=\{
    [U]\in\mathcal{L}^0(V_r) \mid \dist([U],[V_r])\leq r
    \}.
\end{equation}
In general, for each $\ell\in\{0,\ldots,r\}$, we set
\begin{align*}
    \B_\ell([V_r])&=\{[U]\in\mathcal{L}^0(V_r) \mid \dist([U],[V_r])\leq \ell\} \textup{ and }\\
    \partial\B_\ell([V_r])&=\{[U]\in\mathcal{L}^0(V_r) \mid \dist([U],[V_r])= \ell\}, 
\end{align*}
which we call the \emph{ball of radius $\ell$} and the \emph{sphere of radius $\ell$} around $[V_r]$, respectively.

\begin{example}\label{ex:d=2}
Assume that $R=\Z/32\Z$, in which case the maximal ideal of $R$ is generated by $\pi=2$ and $r=5$. \cref{fig:sperner-tree} illustrates the elements of $\mathcal{L}^0(V_5)$: in this picture two elements are joined by an edge if they have distance $1$. We look concretely at some of the elements of $\mathcal{L}^0(V_5)$ and at the distances between them.

If $U$ is the $R$-submodule generated by $4e_1$ and $8e_2$, then $m_U=2$ and $\tilde{U}$ is generated by $e_1$ and $2e_2$. Writing $\gen{X}$ for the $R$-submodule of $V_5$ generated by $X\subseteq V_5$, the tilde representatives of the classes in $\B_1(V_5)$ are
\[ V_5, \quad U_1=\gen{e_1,2e_2}, \quad U_2=\gen{2e_1,e_2}, \quad U_3=\gen{e_1+e_2, 2e_1}\]
while $\partial\B_1(V_5)=\{[U_1],[U_2],[U_3]\}$. Note that $\partial\B_1(V_5)$ is in $1$-to-$1$ correspondence with $\mathbb{P}(V_5/2V_5)$, i.e.\ the elements of $\partial\B_1(V_5)$ can be interpreted as lines in the $2$-dimensional vector space $V_5/2V_5$. Setting now $\ell=2$, we find the representatives of $\partial\B_2(V_5)$:
\[
U_{11}=\gen{e_1,4e_2}, \quad U_{12}=\gen{e_1+2e_2,4e_2}, \quad U_{21}=\gen{e_2,4e_1}, \quad U_{22}=\gen{2e_1+e_2,4e_1}, \] 
\[
U_{31}=\gen{e_1+e_2,4e_1}, \quad U_{32}=\gen{e_1+3e_2,4e_1}.
\]
In the following table we collect the distances within $\B_2(V_5)$:
\begin{center}
    \begin{tabular}{|c|c|c|c|c|c|c|c|c|c|}
    \hline
       & 1 & 2 & 3 & 11 & 12 & 21 & 22 & 31 & 32  \\
       \hline
    1 & 0 & 2 & 2 & 1 & 1 & 3 & 3 & 3 & 3 \\ \hline
     2 && 0 & 2 & 3 & 3 & 1 & 1 & 3 & 3 \\ \hline
    3 &&& 0 & 3 & 3 & 3 & 3 & 1 & 1 \\ \hline
  11 &&&& 0 & 2 & 4 & 4 & 4 & 4 \\ \hline
    12 &&&&& 0 & 4 & 4 & 4 & 4 \\ \hline
  21 &&&&&& 0 & 2 & 4 & 4 \\ \hline
22 &&&&&&& 0 & 4 & 4 \\ \hline
  31 &&&&&&&& 0 & 2 \\ \hline
    \end{tabular}
    \vspace{5pt}
\end{center}
The red dots in \cref{fig:sperner-tree} denote the elements of $\partial\B_3(V_5)$. Moreover, it turns out in this case that $\dist$ on $\mathcal{L}^0(V_5)$ coincides with the graph distance on \cref{fig:sperner-tree}. 
\end{example}

\begin{proposition}\label{prop:ball-sphere-modules}
For each $\ell\in\{0,\ldots,r\}$, the following hold:
\begin{enumerate}[label=$(\arabic*)$]
    \item $\partial\B_{\ell}([V_r])=\{[U]\in\mathcal{L}^0(V_r) \mid |\,[U]\,|= r-\ell+1\}$,
    \item $\B_{\ell}([V_r])=\{[U]\in\mathcal{L}^0(V_r) \mid |\,[U]\,|\geq r-\ell+1\}$.
\end{enumerate}
Moreover, one has $\partial\mathcal{L}(V_r)=\partial\B_r([V_r])$.
\end{proposition}

\begin{proof}
Let $\ell\in\{0,\ldots,r\}$. We start by showing (1). For this, let $U\in\mathcal{L}(V_r)$ and assume without loss of generality that $U=\tilde{U}$. Then the following hold
\begin{align*}
    \dist([U],[V_r])=\ell &\ \Longleftrightarrow \ell=\min\{n\in\Z_{\geq 0} \mid \pi^nV_r\subseteq U \} \\
    &\ \Longleftrightarrow [U]=\{\pi^jU \mid j=0,\ldots,r-\ell\} \\
    &\ \Longleftrightarrow |\,[U]\,|=r-\ell+1
\end{align*}
and so (1) is proven. To show (2), we combine (1) to the observation that 
\[
\B_{\ell}([V_r])=\bigcup_{0\leq j\leq \ell}\partial\B_{j}([V_r])=\bigcup_{r-\ell+1\leq h\leq r+1}\{[U]\in\mathcal{L}^0(V_r)\mid |\,[U]\,|=h\}.
\]
The proof of (1) also shows that $\partial\mathcal{L}(V_r)=\partial\B_r([V_r])$.
\end{proof}

\section{Spherical submodule codes}\label{sec:codes}

\noindent
In this section we define spherical codes in $V_r$ as codes of submodules and prove some initial results. For a comparison with subspace codes see for instance \cite{KSK09} while for a comparison with spherical codes in the Euclidean case, we refer to \cite{EriZin/01}. 

\begin{definition}\label{def:mindist}
Let $\mathcal{C}$ be a subset of $\mathcal{L}^0(V_r)$ with $|\mathcal{C}|\geq 2$. Then the {\em minimum distance} of $\mathcal{C}$ is
\[
\mindist(\mathcal{C})=\min\{\dist([U],[U']) \mid [U],[U']\in\mathcal{C},\ [U]\neq [U']\}.
\]
\end{definition}

\noindent
Recall that $\partial\mathcal{L}(V_r)$ is a metric space equipped with the metric $\dist$ from \cref{sec:module-distance} via the identification in \eqref{eq:boundary-class}.

\begin{definition}\label{def:spherical-code}
A {\em spherical code} in $V_r$ is a subset $\mathcal{C}$ of $\partial\mathcal{L}(V_r)$
with at least $2$ elements.
\end{definition}

\noindent
The terminology ``spherical'' is motivated by \cref{prop:ball-sphere-modules}, from which it follows in particular that 
each element $[U]$ of a spherical code $\Ccal$ satisfies $\dist([U],[V_r])=r$. The proof of the next result is straightforward; compare also with the table in \cref{ex:d=2}.

\begin{lemma}\label{lem:basic-mindist}
For each spherical code $\mathcal{C}$ in $V_r$, one has $\mindist(\mathcal{C})\leq 2r.$
\end{lemma}


\noindent
A spherical code in $V_r$ could in principle equal $\partial\mathcal{L}(V_r)$, so a universal yet weak bound on the cardinality of a spherical code is given by 
$|\partial\mathcal{L}(V_r)|$. For a precise count of the elements of $\partial\mathcal{L}(V_r)$ or $\mathcal{L}^0(V_r)$ we refer to \cref{sec:counting} via \cref{th:isometry}. 
The most interesting bounds for spherical codes come from relating $\dist(\mathcal{C})$ and $|\mathcal{C}|$.

\begin{definition}\label{def:XY}
Let $\chi,\psi$ denote integers satisfying $\chi\geq 2$ and $\psi\geq 1$. Define
\begin{enumerate}
    \item $\dist(d;R;\chi)=\max\{\dist(\mathcal{C}) \mid \mathcal{C}\subseteq \partial\mathcal{L}(V_r),\ |\mathcal{C}|\geq \chi\}$, 
    \item $\card(d;R;\psi)=\max\{|\mathcal{C}| \mid \mathcal{C}\subseteq \partial\mathcal{L}(V_r),\ \dist(\mathcal{C})\geq \psi \}$.
\end{enumerate}
\end{definition}

\noindent
Since (1) and (2) are somewhat dual to each other (see also the analogous definitions in the Euclidean case \cite[Section~2.3]{EriZin/01}), we will mostly be focussing on (2).

\begin{example}\label{ex:d=2-sec3}
 The blue dots in \cref{fig:sperner-tree} form a spherical code in $V_5$ with minimum distance $6$; cf.\ also \cref{ex:d=2}. In particular this shows that $\card(2;\Z/32\Z;6)\geq 12$ and $\dist(2;\Z/32\Z;12)\geq 6$. 
 \end{example}

\begin{definition}
Let $\mathcal{C}=\{[U_1],\ldots, [U_s]\}$ denote an ordered spherical code in $V_r$. The \emph{half-distance matrix} of $\mathcal{C}$ is $N(\mathcal{C})=(n_{ij})\in \Z^{s\times s}$ where 
\[
n_{ij}=
\min\{\beta\in\Z_{\geq 0} \mid \pi^{\beta}\tilde{U_i}\subseteq \tilde{U_j}.\} 
\]
\end{definition}

\noindent
Note that, as a consequence of \cref{prop:ball-sphere-modules}, if $\dist(\mathcal{C})=2r$, then $N(\mathcal{C})=rJ_s$.

\begin{remark}\label{rmk:dist-n_ij}
Let $\mathcal{C}=\{[U_1],\ldots, [U_s]\}$ be an ordered spherical code in $V_r$. 
The \emph{distance matrix} of $\mathcal{C}$ is $$D(\mathcal{C})=N(\mathcal{C})+N(\mathcal{C})^{\mathrm{t}}.$$
Then $D(\mathcal{C})=(\delta_{ij})$ is a symmetric matrix with the following properties:
\begin{enumerate}[label=$(\arabic*)$]
\item for each pair $(i,j)$, one has $\dist([U_i],[U_j])=\delta_{ij}=\delta_{ji}$,
\item $\dist(\mathcal{C})=\min\{\delta_{ij} \mid \delta_{ij}\neq 0\}$.
\end{enumerate}
\end{remark}

\noindent
The following proposition is easily seen to hold as a consequence of \cref{prop:ball-sphere-modules}.

\begin{proposition}\label{prop:2r<->modules}
Let $[U_1],[U_2]$ be in $\mathcal{L}^0(V_r)$. Then the following are equivalent:
\begin{enumerate}[label=$(\arabic*)$]
\item $\dist([U_1],[U_2])=2r$.
\item $[U_1],[U_2]\in\partial\mathcal{L}(V_r)$ and
$\pi^{r-1}U_1, \pi^{r-1}U_2\not\subseteq U_1\cap U_2$.
\end{enumerate}
\end{proposition}


\section{Grassmannians and Sperner codes}\label{sec:Sperner}

\noindent
In this section we build spherical codes in $V_r$ starting from modules highlighted by the investigation of the Sperner property in finite abelian $p$-groups; cf.\ \cite{Spe28,Sta91,Wang98}.
We call a subset $\mathcal{K}$ of a poset $\mathcal{P}$ a \emph{chain} if any two of its elements are \emph{comparable}, i.e.
\[
a,b \in \mathcal{K} \ \Longrightarrow a\preceq b \textup{ or } b\preceq a.
\]
On the contrary, an \emph{antichain} is a subset $\mathcal{A}$ of $\mathcal{P}$ whose elements are pairwise \emph{incomparable}, that is
\[
a,b\in\mathcal{A} \ \Longrightarrow a\not \preceq b \textup{ and } b\not\preceq a.
\]
Antichains play an important role in the construction of ``big codes" in this paper.

\subsection{Grassmannians and Sperner bounds}\label{sec:grass-sperner}

The content of this section could be presented in terms of the Sperner property, though we choose not to do so for the sake of brevity. For the purposes of this section, $\Lcal(V_r)$ is considered as the poset of all  $R$-submodules of $V_r$, ordered by inclusion. Recall that, if $U$ is a free $R$-submodule of $V_r$, then its \emph{rank} equals the minimum cardinality of a generating set.

\begin{definition}
Let $n$ be an integer with $1\leq n \leq d-1$. The \emph{Grassmannian} $\Gr(n,V_r)$ is the collection of all free $R$-submodules of $V_r$ of rank $n$.
\end{definition}

\noindent
It is clear from its definition that, for each $n$, the Grassmannian $\Gr(n,V_r)$ is an antichain in $\Lcal(V_r)$ and is contained in $\partial\mathcal{L}(V_r)$.
Moreover, when $r=1$, the Grassmannian $\Gr(n,V)$ consists of the $n$-dimensional subspaces of $V$. For more on Grassmannians, we refer to \cite[Chapter~5]{MS21} and references therein.
Generalizing the proof of \cite[Proposition~1.3.18]{Sta97} to $V_r$ 
, we have that 
\begin{equation}\label{eq:grass}
    |\Gr(n,V_r)| = \binom{d}{n}_{q^{-1}}q^{rn(d-n)},
\end{equation}
from which it follows that 
$|\Gr(n,V_r)|=|\Gr(d-n,V_r)|$. We remark that \eqref{eq:grass} also follows directly from the more general formulas from \cref{sec:counting}.

\begin{example}\label{ex:grass}
Assume that $d=2$ and $R=\Z/32\Z$, which implies that $q=2$.
We have seen in \cref{ex:d=2} that $\partial\B_1(V_5)$ has the same number of elements as $\Gr(1,V_5/2V_5)$, where $V_5/2V_5$ is viewed as a free $R/2R$-module. Indeed \eqref{eq:grass} ensures
\[
|\Gr(1,V_5/2V_5)| = \binom{2}{1}_{\frac{1}{2}}2^1=\frac{1-\left(\frac{1}{2}\right)^2}{1-\frac{1}{2}}2^1=3=|\partial\B_1(V_5)|.
\]
\end{example}

\noindent
The following is the main result of \cite{Wang98}, which is there phrased to hold for $r\geq 3$. The case where $r=1$ can be found in \cite{RH71,Spe28}, while the case $r=2$ is given in \cite[Theorem~2.7]{Sta91}.

\begin{proposition}\cite[Main Theorem]{Wang98}\label{prop:sperner}
Set $e_-=(d-1)/2$ and $e=d/2$ and $e_+=(d+1)/2$. Let, moreover, $n\in\{1,\ldots,d-1\}$. Then 
$\Gr(n,V_r)$ is a maximal-sized antichain in $\Lcal(V_r)$ if and only if exactly one of the following holds: 
\begin{enumerate}[label=$(\arabic*)$]
    \item $d$ is even and $n=e$,
    \item $d$ is odd and $n\in\{e_-,e_+\}$.
\end{enumerate}
\end{proposition}

\subsection{Sperner codes}\label{sec:sperner-codes}

In this section we define spherical codes in $V_r$ that will yield lower bounds to $\card(d;R;2\alpha)$ for any choice of the integer $1\leq \alpha\leq r$. To this end, we fix such an $\alpha$ and define
\begin{equation}\label{eq:m&e}
m=r+1-\alpha \ \textup{ and } \
e=\lceil d/2 \rceil=\begin{cases}
d/2 & \textup{ if } d \textup{ even},\\
(d+1)/2 & \textup{ if } d \textup{ odd}.
\end{cases}
\end{equation}
We will define a family of codes $\mathcal{C}$ that satisfy
\begin{equation}\label{eq:sperner-dist-card}
\dist(\Ccal)\geq 2\alpha\ \textup{ and }\ \lvert\Ccal\rvert=|\Gr(e,V_{m})|,
\end{equation}
cf.\ \cref{def:sperner-code}.
Write $\Gr(e,\pi^{\alpha-1}V_r)$ for the collection of free $R/\mathfrak{m}^m$-submodules of $\pi^{\alpha-1}V_r$ of rank $e$ and note that the elements of $\Gr(e,\pi^{\alpha-1}V_r)$ are incomparable. Moreover, $\Gr(e,\pi^{\alpha-1}V_r)$ is in bijection with $\Gr(e,V_m)$, equivalently
\[
|\Gr(e,\pi^{\alpha-1}V_r)|=|\Gr(e,V_m)|=\binom{d}{e}_{q^{-1}}q^{me(d-e)}.
\]

\begin{definition}\label{def:sperner-code}
A \emph{Sperner code with parameters} $(d;R;\alpha)$ is a subset $\Ccal$ of $\Gr(e,V_r)$ such that the map
\[
\Ccal \longrightarrow \Gr(e,\pi^{\alpha-1}V_r), \quad U \longmapsto \pi^{\alpha-1}U,
\]
is a bijection.
\end{definition}

\noindent
We remark that a Sperner code with parameters $(d;R;\alpha)$ is nothing else than a collection $\Ccal$ of free $R$-submodules of $V_r$ with the property that $\Gr(e,V_r)=\{\pi^{\alpha-1}U \mid U\in \Ccal\}$. An example of a Sperner code when $d=2$ is given in  \cref{fig:sperner-tree} (see also \cref{ex:d=2,ex:d=2-sec3}).

\begin{theorem}\label{thm:sperner-codes}
Let $1\leq \alpha\leq r$ be an integer and let $\Ccal$ be a Sperner code with parameters $(d;R;\alpha)$. Then the following are satisfied:
\begin{enumerate}[label=$(\arabic*)$]
    \item $\Ccal$ is a spherical code in $V_r$,
    \item $\dist(\Ccal)\geq 2\alpha$,
    \item $\lvert\Ccal\rvert=|\Gr(\lceil d/2 \rceil,V_{r+1-\alpha})|$.
\end{enumerate}
\end{theorem}

\begin{proof}
(1) and (3) are clear from the construction of Sperner codes, so we prove (2). For this,
let $U_1,U_2\in\Ccal$ be distinct: we claim that $n_{12}\geq \alpha$. For a contradiction, assume that this is not the case. It follows that 
$
\pi^{\alpha-1}U_1\subseteq \pi^{n_{12}}U_1\subseteq U_2
$
and so 
\[
\pi^{\alpha-1}U_1\subseteq U_2\cap \pi^{\alpha-1}V_r=\pi^{\alpha-1}U_2,
\]
which contradicts the bijectivity of the map $\Ccal\rightarrow\Gr(e,\pi^{\alpha-1}V_r)$ from \cref{def:sperner-code}. We have proven that $n_{12}\geq \alpha$ and, the choice of $U_1$ and $U_2$ being arbitrary, we have that $\dist(\Ccal)\geq 2\alpha$.
\end{proof}

\noindent
The following is an immediate corollary of the last result.

\begin{corollary}\label{cor:general-bound}
Let $1\leq \alpha\leq r$ be an integer and define $e=\lceil d/2\rceil$. Then
\[
\card(d;R;2\alpha)\geq \binom{d}{e}_{q^{-1}}q^{(r+1-\alpha)e(d-e)}.
\]
\end{corollary}

\noindent
As we will see in the next section, the inequality from \cref{cor:general-bound} is an equality in some cases. We leave the following general question open. 

\begin{question}\label{qs:general-bound}
Is the inequality from \cref{cor:general-bound} always an equality?
\end{question}

\section{Extremal cases}\label{sec:extremal}

\noindent
In this section we show that \cref{qs:general-bound} has a positive answer when $\alpha=r$ or $d=2$ by showing that, in these cases, Sperner codes are optimal codes with respect to the bound given in \cref{cor:general-bound}.


\subsection{Codes of maximal distance}

This section is devoted to the case $\alpha=r$.

\begin{proposition}\label{lem:any-to-free}
Let $\mathcal{C}$ be a spherical code in $V_r$ 
with $\mindist\mathcal{C}=2r$. Then there exists a spherical code $\mathcal{C}'$ in $V_r$ such that the following hold: 
\begin{enumerate}[label=$(\arabic*)$]
\item $|\mathcal{C}|=|\mathcal{C}'|$ and $\mindist(\mathcal{C}')=2r$,
\item\label{it:free} for each $U\in\mathcal{C}'$
, one has $\pi U=U\cap \pi V_r$.
\end{enumerate}
\end{proposition}

\begin{proof}
Define $\Ccal_F$ to be the collection of all $U\in\Ccal$ such that $\pi U=U\cap \pi V_r$.
We prove, by induction on $n=|\mathcal{C}\setminus\mathcal{C}_F|$, that there exists $\mathcal{C}'$ satisfying (1) and (2). If $n=0$, then 
 $\mathcal{C}$ already satisfies (2) and we set $\mathcal{C}'=\mathcal{C}$. Assume now that $n>0$ and that the claim is satisfied for $n-1$.
 Let $U\in\mathcal{C}$ be such that $\pi U\neq U\cap \pi V_r$, that is $U$ is not a free $R$-submodule of $V_r$. 
In view of this, let $X$ and $H$ be submodules of $U$ satisfying
\[
U=X\oplus H, \quad \pi H=H\cap \pi V_{r}, \quad \pi^{r-1}X=0.
\]
In particular, $H$ is isomorphic to the free part (as $R$-submodule) of $U$ and both $H$ and $X$ are non-trivial.
Set now $\mathcal{C}''=(\mathcal{C}\setminus\{U\})\cup\{H\}$. It follows from $\pi^{r-1}U=\pi^{r-1}H$ and \Cref{prop:2r<->modules} that $\mathcal{C}''$ is a spherical code of minimal distance $2r$. Moreover, we have $\mathcal{C}_F''=\Ccal_F\cup\{H\}$ 
and $\mathcal{C}$ and $\mathcal{C}''$ have the same cardinality. We are now done thanks to the induction hypothesis.
\end{proof}

\noindent
Thanks to \cref{lem:any-to-free}, to compute the maximal cardinality of spherical codes of maximal distance in $V_r$ it suffices to look at free $R$-submodules of $V_r$, equivalently at subsets of the sets of vertices of the ball $\B_r$ as a lattice polytope; cf.\ \cref{def:ball,sec:polytopes,sec:spherical-in-BT}.

\begin{definition}
A spherical code $\mathcal{C}$ in $V_r$ is called \emph{free} if it satisfies \Cref{lem:any-to-free}\ref{it:free}.
\end{definition}

\noindent
The next result follows in a straightforward way from \Cref{prop:2r<->modules}.

\begin{lemma}\label{lem:incomparable}
Let $\mathcal{C}$ be a free spherical code in $V_r$ and let $U_1,U_2\in\mathcal{C}$. Then the following are equivalent:
\begin{enumerate}[label=$(\arabic*)$]
\item $\dist([U_1],[U_2])=2r$,
\item $\pi^{r-1}U_1$ and $\pi^{r-1}U_2$ are incomparable.
\end{enumerate}
\end{lemma}


\begin{theorem}\label{th:dist=2r}
Let $e=\lceil d/2\rceil$ be as defined in \eqref{eq:m&e}. Then the following holds:
\[
\card(d;R;2r)=
\binom{d}{e}_{q^{-1}}q^{e(d-e)}.
\]
\end{theorem}

\begin{proof}
Let $\Ccal$ be a spherical code in $V_r$ of maximal cardinality satisfying $\dist(\Ccal)=2r$.
Thanks to \Cref{lem:any-to-free}, we assume without loss of generality that $\mathcal{C}$ is free. Then \Cref{lem:incomparable} yields that the elements of $\mathcal{C}$ are in bijection with a collection of maximal size of incomparable subspaces of $\pi^{r-1}V_r\cong V$. We are now done thanks to \Cref{prop:sperner} and \eqref{eq:grass}. 
\end{proof}

\subsection{Codes in small dimension}\label{sec:d=2}

In this section we answer \cref{qs:general-bound} when $d=2$, which we assume throughout \cref{sec:d=2}. 


\begin{remark}\label{rmk:k-boundary}
There is a number of properties that spherical codes satisfy when $d=2$, which do not generally hold for every spherical code. For instance, each element of $\partial\mathcal{L}(V_r)$ is a free $R$-submodule of $V_r$ and, for every $1\leq \alpha\leq r$, the family $\Gr(1,\pi^{\alpha-1}V_r)$ from \cref{sec:grass-sperner} forms a set of representatives for the classes in $\partial\B_{r+1-\alpha}([V_r])$; cf.\ \cref{fig:sperner-tree}. Write now 
 $\Gr(1,\pi^{\alpha-1}V_r)=\{S_1,\ldots,S_t\}$ and, for each $k\in\{1,\ldots,t\}$, define 
\[
\partial^{k}_{\alpha}\mathcal{L}(V_r)=\{U\in\partial\mathcal{L}(V_r) \mid \dist([U],[S_k])=\alpha-1\}.
\]
Then $\partial\mathcal{L}(V_r)$ equals the disjoint union of the $\partial^{k}_{\alpha}\mathcal{L}(V_r)$'s and defining a Sperner code with parameters $(2;R;\alpha)$ is the same as choosing one element in each $\partial^{k}_{\alpha}\mathcal{L}(V_r)$; cf.\ \cref{fig:sperner-tree}.
\end{remark}

    \begin{figure}[h]
        \centering
        \includegraphics[scale=0.8]{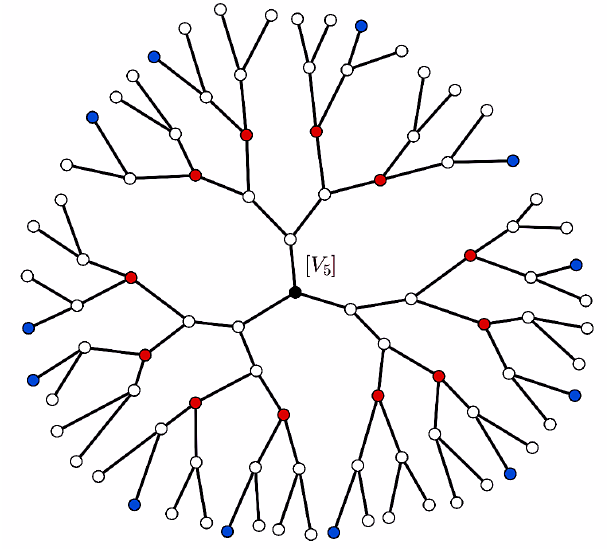}
        \caption{
        The vertices of this graph denote the elements of $\mathcal{L}^0(V_5)$ and $([U],[U'])$ is an edge if $\pi\tilde{U}\subset \tilde{U'}\subset \tilde{U}$ or $\pi\tilde{U'}\subset \tilde{U}\subset\tilde{U'}$.
        The red points are the elements of $\partial\B_3([V_5])$, which are represented by the elements of $\Gr(1,\pi^2V_5)$ from \cref{sec:grass-sperner}.
        A Sperner code with parameters $(2;\Z/32\Z;3)$ and thus minimum distance $6$ is given in blue.
        } 
        \label{fig:sperner-tree}
    \end{figure}

\begin{theorem}\label{th:d=2}
Let $1\leq \alpha\leq r$ be an integer.
Then the following holds:
\[
\card(2;R;2\alpha)=(q+1)q^{r-\alpha}.
\]
\end{theorem}

\begin{proof}
Thanks to \cref{cor:general-bound}, we have that 
\[
\card(2;R;2\alpha)\geq |\Gr(1,V_{r+1-\alpha})|= \binom{2}{1}_{q^{-1}}q^{r+1-\alpha}=(q+1)q^{r-\alpha}, 
\]
so we prove the other inequality. With the notation from \cref{rmk:k-boundary}, we have, for any $k\in\{1,\ldots,t\}$ and $U,U'\in\partial^k_{\alpha}\mathcal{L}(V_r)$, that
\[
\dist([U],[U'])\leq \dist([U],[S_k])+\dist([U'],[S_k])=2(\alpha-1)<2\alpha.
\]
The choice of $k$ being arbitrary, this shows that any spherical code in $V_r$ with $\dist(\Ccal)\geq 2\alpha$, can contain at most one representative from each $\partial^k_{\alpha}\mathcal{L}(V_r)$. This concludes the proof.
\end{proof}

\section{Permutation codes}\label{sec:permutation}

\noindent
In this section, we give a possible generalization of permutation codes, as defined in \cite[Chapter~4]{EriZin/01}, by means of $\Sym(d)$-orbits of $R$-modules with compatible generating sets. For the fixed $R$-basis ${\bf e}=(e_1,\ldots,e_d)$ of $V_r$, we define $\mathcal{L}_{\bf e}(V_r)$ to be the family of $R$-submodules of $V_r$ that can be generated \emph{compatibly with} ${\bf e}$, in other words modules of the form
\[
U_{\delta}=R\pi^{\delta_1}e_1\oplus\ldots \oplus R\pi^{\delta_d}e_d, \textup{ where } 0\leq \delta_i\leq r.
\]
The homothety relation from \cref{sec:module-distance} respects base compatibility and so we define $\mathcal{L}_{\bf e}^0(V_r)$ to be the subfamily of $\mathcal{L}^0(V_r)$ with representatives in $\mathcal{L}_{\bf e}(V_r)$.
In particular, we can model all elements of $\mathcal{L}_{\bf e}^0(V_r)$ in terms of the $\Sym(d)$-orbits of the set $\mathcal{E}_r^{(d)}$ in $\Z^d$ and $\partial\mathcal{L}_{\bf e}(V_r)=\partial\mathcal{L}(V_r)\cap\mathcal{L}_{\bf e}(V_r)$ is defined by permutations of elements of
$\partial \mathcal{E}_r^{(d)}$. 

\begin{example}\label{ex:d=3}
Assume that $d=3$ and $R=\Z/25\Z$, yielding $r=2$ and $q=5$. Then $U_{(0,0,0)}$ is the same as $V_2$ and the modules $U_{(1,1,0)}\subset U_{(1,0,0)}\subset U_{(0,0,0)}$ are pairwise at distance $1$ from each other. Moreover, $U_{(2,1,0)}$ is an element of $\partial\mathcal{L}_{\bf e}(V_2)$. Note that, while $|\partial\mathcal{L}_{\bf e}(V_2)|=12$, the cardinality of $\partial\mathcal{L}(V_2)$ is equal to 1860; cf.\ \cref{sec:counting}. If we compared the spheres of radius $1$ around $[V_2]$, we would get $6$ elements in the compatible case, against the $62$ without basis restrictions. 
\end{example}

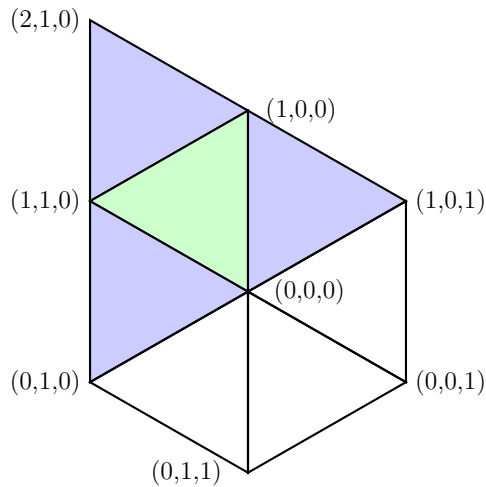
\begin{figure}[h!]
\begin{center}
\begin{tikzpicture}[thick, scale=0.8, every node/.style={scale=0.8}]
 \draw[fill=blue!20] (0,3)  -- (-2.598,1.5)-- (-2.598,4.5) node[anchor=east] {(2,1,0)} -- cycle;
 \draw (-0.9,3); 
    \draw[fill=green!20] (0,3)node[anchor=west] {\ (1,0,0)}  -- (-2.598,1.5)node[anchor=east] {(1,1,0)} -- (0,0) node[anchor=west] {\ \ (0,0,0)} -- cycle;
    \draw (-0.5,1.5); 
      \draw[fill=blue!20] (0,3) -- (2.598,1.5)node[anchor=west] {(1,0,1)} -- (0,0)  -- cycle;
    \draw (0.4,1.5); 
    \draw[fill=blue!20] (0,0) -- (-2.598,1.5) -- (-2.598,-1.5) node[anchor=east]{(0,1,0)} --cycle;
    \draw (-0.9,0); 
    \draw (0,0) --  (-2.598,-1.5) -- (0,-3) node[anchor=east]{(0,1,1)\ \ \ } --cycle;
    \draw (0,0)  -- (0,-3) --  (2.598,-1.5) node[anchor=west]{(0,0,1)} --cycle;
      \draw (0,0)  --   (2.598,-1.5) -- (2.598,1.5) --cycle;
  \end{tikzpicture}
  \caption{A local picture of $\mathcal{L}^0_{\bf e}$ when $d=3$.}\label{figureA2}
  \end{center}
\end{figure}

	\subsection{Tropical operations and polytropes}\label{sec:polytopes}
	
For the sake of conciseness and in adherence to the references cited below we introduce here some more notation, coming from tropical geometry.
For real elements $a$ and $b$ we set
	$$  a \minplus b  = {\rm min}\{a,b\}, \quad 
	a \maxplus b  = {\rm max}\{a,b\}, \quad
	a \odot b = a + b$$
	and remark that the last operations can be extended to $\R^d$ componentwise. 
	For each matrix $M\in\R^{d\times d}$ with $0$'s on the diagonal, we define moreover
	\begin{equation}
		\label{eq:polytrope}
		Q(M)  \,\, = \,\, \bigl\{ u \,\in \R^d/\R{\bf 1} \,:\, u_i - u_j \,\, \leq \, m_{ij} \,\,\,
		\hbox{for} \,\, 1 \leq i,j \leq d \,\bigr\}, 
	\end{equation}
which is a convex polytope in $\R^{d}/\R{\bf 1}$ and is called a {\em polytrope} in tropical geometry. For more on polytropes, we refer the interested reader to \cite{EHNSS/21,Jos22,joskul, maclagan}. In this paper, we will only deal with polytropes like the ones in the next example. As we mention in \cite[Example~13]{EHNSS/21}, such polytropes are called \emph{pyropes} in \cite{joskul} and can be seen as balls of radius $r$ in the tropical metric \cite[Section~3.3]{CGQ04}. Recall that $J_d$ denotes the matrix in $\Z^{d\times d}$ with $0$'s on the diagonal and off-diagonal entries all equal to $1$.

\begin{example}\label{ex:Jd}
Let $[\delta]\in Q(rJ_d)$ be such that $\delta$ has integral coordinates. Then there exists $\tilde{\delta}\in[\delta]$ all of whose coordinates $\tilde{\delta}_i$ are integral and satisfy $0\leq\tilde{\delta}_i\leq r$. Then $U_{\tilde{\delta}}$ belongs to $\mathcal{L}_{\bf e}(V_r)$ and, any other $\tilde{\delta}'$ such that 
\[
\tilde{\delta}'\in \tilde{\delta}+\Z{\bf 1}\ \textup{ and }\ 0\leq\tilde{\delta}_i'\leq r
\]
yields $[U_{\tilde{\delta}}]=[U_{\tilde{\delta}'}]$. More precisely, using the language of buildings, one can show that there is a one-to-one correspondence between the integral points of $Q(rJ_d)$ and the elements of $\mathcal{L}_{\bf e}^0(V_r)$; cf.\ \cref{th:isometry} and  \cite[Theorem~5.2]{EGS/21}. 
\end{example}	

\noindent	
Identifying $\R^d/\R{\bf 1}$ with $\{u\in\R^d \mid u_d=0\}$, it is not
difficult to see from \cref{eq:polytrope}  that the coordinates of vertices of the polytope $Q(rJ_d)$ are in $\{0,r\}^d\cup\{0,-r\}^d$. As mentioned in the Introduction, this has a nice interpretation in terms of free $R$-submodules of $V_r$.
	
\subsection{Permutation codes}	In this section we define permutation codes and give examples of such codes in connection with the theory of polytropes. In \cref{th:permutation} we give sharp bounds on the minimum distance and cardinality of permutation codes in terms of their defining parameters. 

\begin{definition}
An \emph{${\bf e}$-permutation code} in $V_r$ is a code of the form
\begin{equation}\label{eq:perm-codes}
\mathcal{C}=\{[U_{\delta}] \mid \delta\in\Sym(d)\cdot \varepsilon\}, \textup{ where } \varepsilon\in \partial \mathcal{E}_r^{(d)}.
\end{equation}
\end{definition}

\noindent
To lighten the notation, we will often write $\mathcal{C}=\Sym(d)\cdot\varepsilon$ for a code as in \eqref{eq:perm-codes}.

\begin{remark}\label{rmk:trop-dist}
Let $\delta,\varepsilon \in\{0,\ldots, r\}^d$. Then the distance between $[U_{\delta}]$ and $[U_{\varepsilon}]$ is given by 
\[
\dist([U_{\delta}],[U_{\varepsilon}])=\max_{i=1,\ldots,d}\{\delta_i-\varepsilon_i\}-\min_{i=1,\ldots,d}\{\delta_i-\varepsilon_i\}.
\]
This can be proven by direct computation or relying on \cref{th:isometry} and \cite[Remark~3.3]{EGS/21}.
\end{remark}

\begin{remark}
One could replace $\Sym(d)$ with $\Aut(V_r)$ and study codes of the form $\Aut(V_r)\cdot \varepsilon$, that is maximal codes consisting of pairwise isomorphic $R$-modules. However, one can already see for $d=2$ that these codes are not particularly interesting in terms of general bounds. More precisely, if $d=2$, one has $\partial\mathcal{L}(V_r)=\Aut(V_r)\cdot (r,0)=\mathcal{C}$ and so $\dist(\mathcal{C})=2$ while $|\mathcal{C}|=|\Gr(1,V_r)|=(q+1)q^{r-1}$.
\end{remark}

\noindent
Of particular interest are codes that are derived from vertices of the polytrope $Q(rJ_d)$; cf.\ \cref{ex:Jd}. Such vertices are given by permutations of elements $\varepsilon$ of $\mathcal{E}_r^{(d)}$ whose entries satisfy $\{0,r\}=\{\varepsilon_1,\ldots,\varepsilon_d\}$, in other words they correspond to the free $R$-submodules of $V_r$.
For each $n\in\{1,\ldots,d-1\}$, we set
\[
\mathcal{F}_r^n=\Sym(d)\cdot (\underbrace{r,\ldots,r}_{d-n},0,\ldots,0)
\]
describing the collection of all free $R$-submodules of $V_r$ that belong to $\mathcal{L}_{\bf e}(V_r)$.
Note that, by its definition, each $\mathcal{F}_r^n$ is contained in $\partial\mathcal{L}(V_r)$ 
and the cardinality of $\mathcal{F}_r^n$ is equal to 
\[
|\mathcal{F}_r^n|=\binom{d}{n}=\frac{d!}{n!(d-n)!}.
\]

\begin{example}
In \Cref{fig:rhombic}, the 14 regular vertices of the polytope $Q(J_4)$ are so divided: 
\begin{itemize}
\item the red vertices 
describe $\mathcal{F}_1^1$, 
\item the blue vertices 
describe $\mathcal{F}_1^3$, 
\item all other vertices, i.e.\ the yellow ones, are the elements of $\mathcal{F}_1^2$.
\end{itemize}
Moreover,  in the language of tropical geometry, the red and blue vertices are the  min- resp.\ max-vertices of the polytrope $Q(J_4)$; cf.\ \cite[Example~1,Theorem~16]{EHNSS/21}.
\end{example}

\begin{figure}[h!]
\includegraphics[scale=0.4]{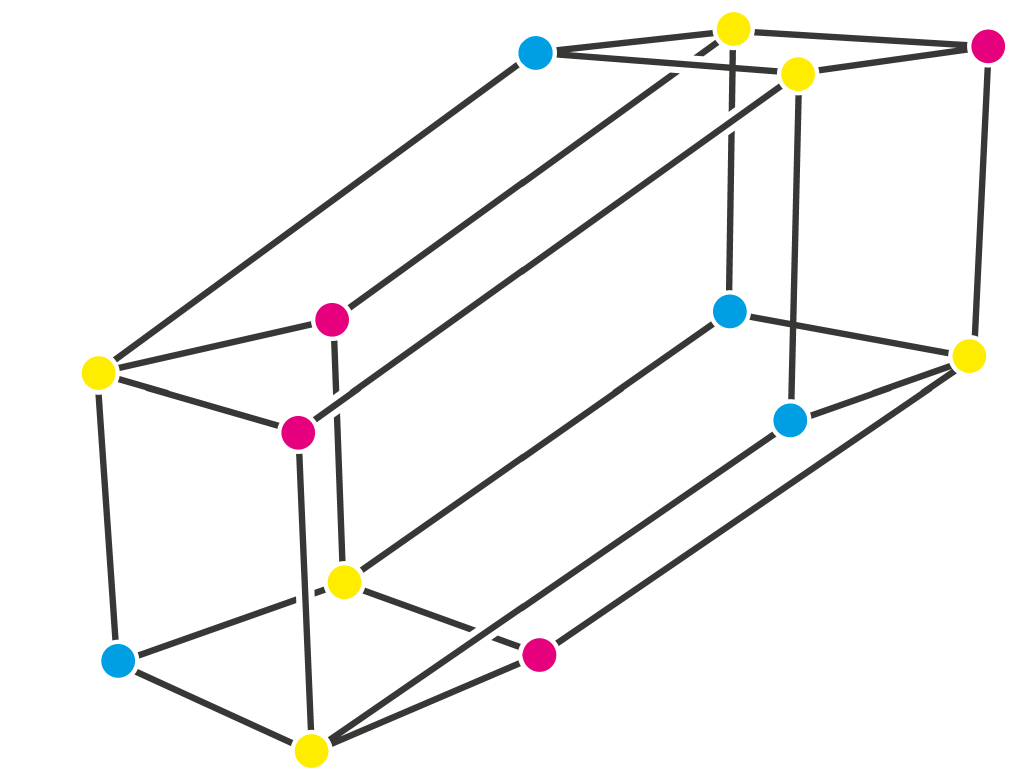}
\caption{A representation in $\R^4/\R{\bf 1}$ of $Q(J_4)$. The yellow dots constitute an ${\bf e}$-permutation code of maximal size having distance $2$; cf.\ \cref{th:permutation}.}\label{fig:rhombic}
\end{figure}

\noindent
In the following results we  
compute cardinality and minimal distance of permutation codes. 
For this, We fix $\varepsilon\in\partial\mathcal{E}_r^{(d)}$ and write
\begin{itemize}
    \item $\ell=-1+\lvert \{\varepsilon_1,\ldots,\varepsilon_{d}\} \rvert \geq 0$, 
    \item $\{\varepsilon_1,\ldots,\varepsilon_{d}\}=\{\tilde{\varepsilon}_1>\ldots >\tilde{\varepsilon}_{\ell+1}=0\}$.
\end{itemize}
For each $s\in\{1,\ldots, \ell+1\}$, we define moreover
$	m_s=\lvert \{i\in\{1,\ldots,d\} \mid \varepsilon_i=\tilde{\varepsilon}_s\} \rvert$ and note that $m_1+\ldots+m_{\ell+1}=d$.

\begin{proposition}\label{prop:dist-permutation}\label{lem:card-permutation}
For $\mathcal{C}=\Sym(d)\cdot \varepsilon$, the following hold:
\[
|\mathcal{C}|=\frac{d!}{m_1!m_2!\cdots m_{\ell+1}!}\ \textup{ and }\
\dist(\mathcal{C})=2\bigminplus_{1\leq j<i\leq \ell+1}\tilde{\varepsilon}_i-\tilde{\varepsilon}_j.
\]
\end{proposition}

\begin{proof} 
The first equality follows straightforward from the definition, so we prove the second.
To this end, write $\Ccal=\Sym(d)\cdot \varepsilon$ and set 
Set $[U_1]=[U_{\varepsilon}]$.
Let, moreover, $[U_2]\in\mathcal{C}$. In view of \cref{rmk:trop-dist},
to minimize $\dist([U_1],[U_2])$ we pick indices $h,k\in\{1,\ldots,d\}$ such that 
\[
\varepsilon_h-\varepsilon_k=\bigminplus_{1\leq j<i\leq \ell+1}\tilde{\varepsilon}_i-\tilde{\varepsilon}_j
\]
and define $\sigma$ to be the transposition in $\Sym(d)$ interchanging $h$ and $k$. Choosing $[U_2]$ to correspond to $\sigma\cdot \varepsilon$, we get from \cref{rmk:trop-dist} that 
\[
\dist(\mathcal{C})=\dist([U_1],[U_2])=2(\varepsilon_h-\varepsilon_k)=2\bigminplus_{1\leq j<i\leq \ell+1}\tilde{\varepsilon}_i-\tilde{\varepsilon}_j.
\]
\end{proof}

\begin{corollary}\label{th:free}
Let $n\in\{1,\ldots,d-1\}$. Then $\mathcal{F}_r^n$ is a spherical code in $V_r$ of minimal distance $2r$.
\end{corollary}

\noindent
In the following result, we provide sharp bounds for minimum distance and cardinality when cardinality and minimum distance are given, respectively.

\begin{theorem}\label{th:permutation}
Let $1\leq\alpha\leq r$ be an integer and write $\mathcal{C}=\Sym(d)\cdot \varepsilon$.
Then the  following are satisfied:
\begin{enumerate}[label=$(\arabic*)$]
\item If $r=\delta \ell+Z$ with $\delta, Z$ non-negative integers satisfying $Z<\ell$, then 
\[\dist(\mathcal{C})\leq 2\delta.\]
\item Write $r=\alpha X+Y$ and $d=\beta X+\gamma$, for $X,Y,\beta,\gamma$ non-negative integers satisfying $Y<\alpha$ and $\gamma<X$.
If $\dist(\mathcal{C})=2\alpha$, then 
\[
|\mathcal{C}|\leq \frac{d!}{(\beta!)^{X+1}(\beta+1)^\gamma} .
\]
\end{enumerate}
\end{theorem}

\begin{proof}
We start by proving (1). For this, write $r=\delta\ell+Z$ and assume without loss of generality that $\varepsilon$ is such that $\dist(\Ccal)$ is maximal. Thanks to \cref{prop:dist-permutation}, maximizing the minimum distance of $\Ccal$ is the same as maximizing the minimum $\eta$ of the set $\{\tilde{\varepsilon}_i-\tilde{\varepsilon}_j \mid i>j\}$. This is clearly achieved for $\eta=\delta$. We now prove (2). To this end, assume that $\dist(\Ccal)=2\alpha$. 
Thanks to \cref{prop:dist-permutation}, we know that $\min\{\tilde{\varepsilon}_i-\tilde{\varepsilon}_j \mid i>j\}=\alpha$ and we now need to determine $\varepsilon$ for which 
\[
|\mathcal{C}|=\frac{d!}{m_1!m_2!\cdots m_{\ell+1}!}
\]
is maximal, i.e.\ for which $m_1!m_2!\cdots m_{\ell+1}!$ is minimal. This happens when $\ell$ is as large as possible and the $m_i$'s are all roughly the same (i.e. the same or differing by $1$). In view of this, $\ell=X$ and $\{m_1,\ldots,m_{X+1}\}\subseteq\{\beta,\beta+1\}$. More precisely, the number of $m_i$'s that are equal to $\beta+1$ is $\gamma$ and so \cref{lem:card-permutation} yields
\[
|\Ccal|\leq \frac{d!}{(\beta!)^{X+1-\gamma}((\beta+1)!)^{\gamma}}=\frac{d!}{(\beta!)^{X+1}(\beta+1)^\gamma}.
\]
This concludes the proof.
\end{proof}

\noindent
It is not difficult to see, from the proof of \cref{th:permutation}, how one can build optimal codes in this context, i.e.\ permutation codes achieving the bounds from \cref{th:permutation}. As for the case of regular spherical codes, optimal permutation codes are not unique; cf.\ \cref{sec:sperner-codes}. Another thing that is worth mentioning is that optimal permutation codes are far from being optimal in the sense of \cref{cor:general-bound}. We remark that, similar bounds to those of \cref{th:permutation} are proven for a different type of permutation codes in \cite{MN20}.


\begin{question}
Are there other interesting generalizations of permutation codes in the context of buildings? What about group codes; cf.\ \cite[Chapter~8]{EriZin/01}?
\end{question}

\noindent
In the following remark we stress how, in terms of storage and decoding, permutation codes stand out among spherical codes (in accordance with the Euclidean setting).

\begin{remark}[A note on storage and decoding]\label{rmk:decoding}
Let $\mathcal{C}$ be any spherical code in $V_r$. Then the elements of $\mathcal{C}$ can be encoded in a vector of $(d\times d)$-matrices with coefficients in $R$ where the row-span of each matrix identifies an element $[U]$ of $\Ccal$ via returning its $\tilde{U}$ representative. A convenient choice would be to communicate these matrices in row echelon form. In the special case when $\Ccal$ is a permutation code, it however suffices to store an element of $\partial\mathcal{E}_r^{(d)}$ to give full information on the code $\Ccal$.

For what concerns decoding, the lack of additional structure makes it difficult to give a straightforward algorithm for the decoding of general spherical codes of modules, even in the case where they are known to be Sperner codes.
However, thanks to \cref{rmk:trop-dist} and in agreement with the Euclidean case, the decoding of ${\bf e}$-permutation codes is relatively simple. To illustrate this, we fix a permutation code $\Ccal$ and a vector $\eta=(\eta_1,\ldots,\eta_d)\in\mathcal{E}_r^{(d)}$ (note that actually $\eta$ can be taken in $\Z^d$ as the following algorithm allows us also to work in balls of larger radius; cf.\ \cref{sec:dist-balls-buildings}). 
We want to find $\varepsilon^*\in \partial\mathcal{E}_r^{(d)}$ such that $U_{\varepsilon^*}\in\Ccal$ and $\dist([U_{\eta}],[U_{\varepsilon^*}])=\dist([U_{\eta}],\Ccal)$. We follow the steps below: 
\begin{enumerate}
    \item Let $\eta'\in\Z^d$ and $\sigma\in\Sym(d)$ be such that $\eta_1'\geq \ldots \geq \eta_d'$ and $\eta'=\sigma\cdot \eta$.
    \item Define $\tilde{\eta}=\eta'-\eta_d'{\bf 1}$.
    \item Choose $\varepsilon\in\Ccal$ and identify $h,k\in\{1,\ldots, d\}$ such that $\varepsilon_h=0$ and $\varepsilon_k=r$.
    \item Define $\varepsilon'$  and $\tau\in\Sym(d)$ to satisfy 
    $\varepsilon'=\tau\cdot\varepsilon$ and $\varepsilon_1'=0$ and $\varepsilon_d'=r$. 
    \item Set $\varepsilon^*=\sigma^{-1}\cdot \varepsilon'$.
\end{enumerate}
We see from its construction that the element $\varepsilon^*$ might not be unique. It is, however, not difficult to design an algorithm avoiding choices, once $\Ccal$ is given.
\end{remark}

\section{Spherical codes in Bruhat--Tits buildings}\label{sec:spherical-in-BT}

\noindent
In this section we rephrase the results of this paper in terms of buildings. As we will see, Bruhat--Tits buildings are a way of talking about \emph{lattices} and via these objects we can consider balls (in the sense of \cref{sec:module-distance}) ``of any radius'' at the same time. Moreover, it is worth mentioning that, on top of their central role in the theory of reductive groups, buildings have many different  applications, for instance in optimization \cite{BFGOWW19,HH21}, statistics \cite{EM21,ET19}, and coding theory \cite{LNVC18}. Though the employment of buildings in the study and construction of codes is not new, this seems to be the first time spherical codes in buildings are considered. In the applications of flags to network coding, spherical buildings are used. Such strategy, first introduced in \cite{LNVC18}, has found further developments in \cite{AGNP21,AGNPSE20,AGNPSE21,Kurz}
and variations in \cite{FN21}. Moreover, Bruhat--Tits buildings 
 also make their appearance in the study of holographic codes \cite{Mar18} as well as in the study of valued rank-metric codes \cite{EHNS/21}. 
 
 \subsection{From chain rings to valued fields}\label{sec:notation+}

\noindent
We choose a discretely valued field $(K,\val)$, with valuation ring $\Ocal_K$, uniformizer $\pi$, and unique maximal ideal $\mathfrak{m}_K=\Ocal_K\pi\neq 0$, in such a way that 
$R\cong \Ocal_K/\mathfrak{m}_K^r$; cf.\ \cite[\S 1]{AA/22}.
With a slight abuse of notation, we set ${\bf e}=(e_1,\ldots,e_d)$ to be the standard basis of $K^d$ and we write $\Ocal_K^d$ for the free $\Ocal_K$-module $\Ocal_K^d=\Ocal_Ke_1\oplus\ldots\oplus\Ocal_Ke_d$.
We will use the bar notation for the subobjects of $V_r$: if $L\subseteq\Ocal_K^d$, then $\overline{L}$ denotes the image of $L$ in $V_r$ under the natural projection $\Ocal_K^d\rightarrow V_r$. Up to very small variations, our notation is compatible with the one from \cite{EGS/21}.

\subsection{Lattices and buildings}

An \emph{$\Ocal_K$-lattice} (or simply lattice) in $K^d$ is a free $\Ocal_K$-submodule of maximal rank $d$. The (\emph{homothety}) \emph{class} of a lattice $L$ in $K^d$ is
 \[
 [L]=\{cL \mid c\in K\setminus\{0\}\}=\{\pi^nL \mid n\in\Z\},
 \]
  while $\End_{\Ocal_K}(L)$ denotes the endomorphism ring of $L$ as an $\Ocal_K$-submodule of $K^d$, i.e.\ the collection of $\Ocal_K$-linear maps $K^d\rightarrow K^d$ that stabilize $L$. 
  Note that any two homothetic lattices have the same endomorphism ring. Moreover, lattices in $K^d$ form one orbit under the natural action of $\GL_d(K)$ and so it will often not be restrictive to assume (up to base change) that a given lattice $L$ is equal to $\Ocal_K^d$. Additionally, each element of $\mathcal{L}(V_r)$ can be obtained from a lattice $\pi^r\Ocal_K^d \subseteq L\subseteq \Ocal_K^d$, via  projecting $L$ to $V_r$:
  \begin{equation}\label{eq:projection}
  \Ocal_K^d \supseteq L \longmapsto \overline{L}=U_L \subseteq V_r.
  \end{equation}
 We stress that the notions of equivalence for lattices and modules are compatible by means of the last projection.
In line with the content of this paper, we define the affine building of $\SL_d(K)$ via its lattice class model \cite{AbramenkoNebe,Garret} and refer the interested reader to \cite{AbramenkoBrown} for the more general description.

\begin{definition}\label{def:building}
The \emph{affine building} $\Bcal _d(K)$ is an infinite simplicial
complex such that  
\begin{enumerate}[label=$(\arabic*)$]
    \item the vertex set is $\mathcal{B}_d^0=\{ [L] \mid L \mbox{ is an $\Ocal _K$-lattice in } K^{d} \}. $
    \item $\{ [L_1], \ldots , [L_s] \} $ is a simplex
in $\Bcal _d(K)$ if and only if, up to permutation of the indices and choice of representatives, one has $L_1 \supset L_2 \supset \cdots \supset L_s \supset \pi L_1$.
\end{enumerate}
The {\em standard apartment} of $\mathcal{B}_d(K)$ is the subset $\mathcal{A}$ of $\mathcal{B}^0_d(K)$ of all lattice classes with representatives of the form 
\[
L_u=\Ocal_K\pi^{u_1}e_1\oplus\ldots\oplus\Ocal_K\pi^{u_d}e_d, \textup{ where } u=(u_1,\ldots,u_d)\in\Z^d.
\]
\end{definition}
\noindent
More generally, one could define an apartment for any frame choice in $K^d$, cf.\ \cite[Section~2]{EGS/21}. 
Since \eqref{eq:projection} respects homothety classes, 
the new terminology allows us to consider the codes from \cref{sec:permutation} as one-apartment codes in buildings.

\begin{example}\label{ex:buildings}
The rings from \cref{ex:d=2} and \cref{ex:d=3} can both be expressed as quotients of a $p$-adic ring: in the first case $R\cong\Ocal_K/\mathfrak{m}_K^r=\Z_2/(2\Z_2)^5$ while in the second case $R\cong\Ocal_K/\mathfrak{m}_K^r=\Z_5/(5\Z_5)^2$. When $d=2$ or $d=3$, local pictures of $\B_d(\Q_2)$ can be found in \cite[Figures~2-5]{BekSol22}. 
\end{example}

\subsection{Distance and balls}\label{sec:dist-balls-buildings}
	
The following distance was introduced in \cite[Definition~3.1]{EGS/21}. In view of \cref{th:isometry}, we use the same notation as in \cref{def:dist-mod}.

\begin{definition}\label{def:dist}
	Let $[L_1], [L_2] \in {\mathcal B}_d^0(K)$ be two 
	homothety classes of lattices. 
	Then 
$$\dist([L_1],[L_2])= \min \{ s \mid 
	\mbox{ there are } L_1'\in [L_1], L_2' \in [L_2] \mbox{ with } 
	\pi ^{s } L_1' \subseteq L_2' \subseteq L_1' \} .$$
	\end{definition}
	
\noindent 
As proven in \cite[Lemma~3.2]{EGS/21}, the map $\dist:\mathcal{B}_d^0(K)\times\mathcal{B}_d^0(K)\rightarrow\Z$ defines a distance on $\mathcal{B}_d^0(K)$. In view of this, it makes sense to define balls in $\mathcal{B}_d^0(K)$.

	\begin{definition} \label{def:ball} 
		Let $[L]$ be a lattice class in $\mathcal{B}_d^0(K)$. Then  the {\em (closed) ball of radius $r$ and center $[L]$} is
		\[\B_r([L])= \{ [L'] \in {\mathcal B}_d^0(K) \mid 
		\dist ([L],[L'] ) \leq r \} \]
and its \emph{boundary} is
\[
\partial\B_r([L])=\B_r([L])\setminus \B_{r-1}([L])=\{[L']\in\mathcal{B}_d^0(K) \mid \dist([L'],[L])=r\}.
\]
If $[L]=[\Ocal_K^d]$, we write simply $\B_r$ and $\partial\B_r$  for $\B_r([\Ocal_K^d])$ and  $\partial\B_r([\Ocal_K^d])$, respectively.
	\end{definition} 
	
	\begin{example}
 Assume $d=2$. Then $\mathcal{B}_2(K)$ is a $(q+1)$-regular tree and $\dist$ equals the graph distance on $\mathcal{B}_2(K)$. \cref{fig:sperner-tree} represents $\B_5$ as a subset of $\mathcal{B}_2(\Q_2)$. In the same figure, the red points constitute $\partial\B_3$. For more on buildings as trees, see for instance \cite{SerreTrees}.
\end{example}
	
	\noindent
	Balls in the affine building $\mathcal{B}_d(K)$ naturally arise as the collections of stable lattice classes of ball orders \cite[Section~5]{EGS/21} and can be modeled by means of the submodules of $V_r$.
	
	\begin{theorem}\label{th:isometry}
	The following are isometric:
	\begin{enumerate}[label=$(\arabic*)$]
	    \item $\mathcal{L}^0(V_r)$ and $\B_r$,
	    \item $\partial\mathcal{L}(V_r)$ and $\partial\B_r$,
	    \item $\mathcal{L}^0_{\bf e}(V_r)$ and $\B_r\cap \mathcal{A}$,
	    \item $\partial\mathcal{L}_{\bf e}(V_r)$ and $\partial\B_r\cap\mathcal{A}$.
	\end{enumerate}
	\end{theorem}
	
	\begin{proof}
We show (1). To this end, we start by observing that $[L]\in\B_r$ if and only if there exists a representative $L'\in[L]$ such that $\pi^r\Ocal_K\subseteq L'\subseteq \Ocal_K^d$. Since \cref{eq:projection} respects homothety, it is clear that $\B_r$ and $\mathcal{L}^0(V_r)$ are in bijection via $\Ocal_K^d\rightarrow V_r$. We show that the distances are also compatible. For this, let $\pi^r\Ocal_K\subseteq L_1,L_2\subseteq \Ocal_K^d$ be lattices and write $U_1=\overline{L_1}$ and $U_2=\overline{L_2}$.
Assume without loss of generality that $U_1=\tilde{U_1}$ and $U_2=\tilde{U_2}$.
Set $\alpha=\dist([L_1],[L_2])$ and let $n_{12}$ and $n_{21}$ be as in \cref{def:dist-mod}. 
It follows from the definitions of $U_1$ and $U_2$ that $L_1\supseteq \pi^{n_{21}}L_2\supseteq \pi^{n_{21}}(\pi^{n_{12}}L_1)=\pi^{n_{21}+{n_{12}}}L_1$ and in particular $\alpha\leq n_{12}+n_{21}$. 
Without loss of generality, let now $m$ be a non-negative integer such that $L_1\supseteq \pi^mL_2\supseteq \pi^{\alpha} L_1$. Then it follows from the definitions of $n_{21}$ and $n_{21}$ that $m\geq n_{21}$ and $\alpha-m\geq n_{12}$. Moreover, we have
\[
\pi^{n_{12}+n_{21}}L_1\subseteq \pi^{\alpha}L_1\subseteq \pi^mL_2, 
\]
which in turn yields that $\pi^{n_{12}+n_{21}-m}L_1\subseteq L_2$. It follows from the definition of $n_{12}$ that $m=n_{21}$ and thus we derive that $\alpha\geq n_{12}+n_{21}$.
This proves (1) and so, as a consequence, also (2),(3), and (4).	
	\end{proof}
	
\noindent
In view of the last theorem, we transport \cref{def:mindist,def:spherical-code} to the framework of Bruhat--Tits buildings.
	
\begin{definition}\label{def:spherical-code-building}
A \emph{spherical code} in $\B_r$ is a subset $\Ccal$ of $\B_r$ with $|\Ccal|\geq 2$. The \emph{minimum distance} of $\Ccal$ is 
\[
\dist(\Ccal)=\min\{\dist([L_1],[L_2]) \mid [L_1],[L_2]\in\Ccal,\ [L_1]\neq [L_2]\}.
\]
\end{definition}

\noindent
The results from \cref{sec:Sperner,sec:extremal,sec:permutation} can now be also stated in terms of spherical codes in buildings. We close this section with a connection to an earlier paper.
The following is the same as \cite[Definition~5.5]{EGS/21}.

\begin{definition}
	A {\em star configuration} $\star_r ([L]) $ with center $[L]$ and
	radius $r$ is a set 
	$$\star _r ([L]) = \{ [L_1],\ldots , [L_d], [L_{d+1}] \} $$ 
	such that the following hold:
	\begin{enumerate}[label=$(\arabic*)$]
	    \item $\pi ^r L \subseteq L_1,\ldots , L_{d+1} \subseteq L$,
	    \item for each $i\in\{1,\ldots,d+1\}$, one has $L_i / \pi^r L \cong R$,
	    \item for each $i\in\{1,\ldots,d+1\}$, one has $L= \sum _{j\neq i} L_j $. 
	\end{enumerate}
\end{definition}

\begin{proposition}\label{prop:star-is-spherical}
A star configuration $\star([\Ocal_K^d])$ with center $[\Ocal_K^d]$ and radius $r$ is a spherical code with $\mindist(\star([\Ocal_K^d]))=2r$. 
\end{proposition}

\begin{proof}
Write $\star([\Ocal_K^d])=\{[L_1],\ldots,[L_{d+1}]\}$.  In view of conditions (1)-(2)-(3) above, up to a convenient base change, we assume without loss of generality that
\[
L_i=\begin{cases}
\Ocal_Ke_i+\pi^r\Ocal_K^d & \textup{ if } 1\leq i\leq d, \\
\Ocal_K(e_1+\ldots+e_d)+\pi^r\Ocal_K^d & \textup{ if } i=d+1.
\end{cases}
\]
It is clear that $\star([\Ocal_K^d])$ is a spherical code in $\B_r$.
Fix now $i\neq j$. Then \Cref{th:isometry,prop:2r<->modules} yield $\dist([L_i],[L_j])=2r$ and, the choice of $i,j$ being arbitrary, it follows that $\dist(\star([L]))=2r$.
\end{proof}

\noindent
The next corollary follows in a straightforward way from \Cref{def:XY}, with the combination of \Cref{lem:basic-mindist,th:isometry,prop:star-is-spherical}.

\begin{corollary}
One has $\dist(d;R;d+1)=2r$.
\end{corollary}

\section{Counting elements of balls}\label{sec:counting}

\noindent
This section is meant to add to the understanding of balls of modules, resp.\ balls in buildings in terms of their elements' count. The results of this section are self contained and do not explicitly extend results from previous sections though they call for some new observations and questions; cf.\ \cref{rmk:asymptotics,rmk:sphere-pack-bounds}.

We work here under the assumption of \cref{sec:notation+}, though we do not necesarily assume that the residue field of $K$ is finite.
We leverage on results from \cite{Voll10}, in particular its Section~3,  to give a polynomial counting the lattice classes in the ball $\B_r$. More precisely, we define $b^{(d)}_r(X)\in\Z[X]$ such that, if  $q=\lvert  \Ocal_K/\mathfrak{m}_K  \rvert $ is finite, then $\lvert \B_r \rvert =b^{(d)}_r(q)$. We do so by writing $b^{(d)}_r(X)=\sum_{\varepsilon\in\mathcal{E}_r^{(d)}}b_{\varepsilon}(X)$ where, contrarily to what is done in \cref{sec:permutation}, here $\mathcal{E}_r^{(d)}$ parametrizes the elementary divisor types of lattices  $\pi^r\Ocal_K^d\subset L\subseteq \Ocal_K^d$ up to homothety.
	The role of the polynomial $b_{\varepsilon}$ will be to count all lattice classes with the same elementary divisors. We fix $\varepsilon\in\mathcal{E}_r^{(d)}$ and proceed to define $b_{\varepsilon}(X)$. For this, write $\ell=-1+\lvert \{\varepsilon_1,\ldots,\varepsilon_{d}\} \rvert \geq 0$ and $\{\varepsilon_1,\ldots,\varepsilon_{d}\}=\{\tilde{\varepsilon}_1>\ldots >\tilde{\varepsilon}_{\ell+1}=0\}$.  	Now, for each $s\in\{1,\ldots, \ell\}$, define 
	\[
	i_s=\lvert \{i\in\{1,\ldots,d\} \mid \varepsilon_i\geq\tilde{\varepsilon}_s\} \rvert \ \textup{ and}\ 
	r_{i_s}=\tilde{\varepsilon}_s-\tilde{\varepsilon}_{s+1}.
	\]
	 We set, moreover  $\Lambda_{\varepsilon}=\Ocal_K^{d\times d}\cap\End_{\Ocal_K}(L_{\varepsilon})$ and $I=I(\varepsilon)=\{i_1<\ldots <i_{\ell}\}$.  In terms of these parameters, the endomorphism ring $\End_{\Ocal_K}(L_{\varepsilon})$ is denoted $\Gamma_{I,{\bf r}}$ in \cite{Voll10} and is explicitly described in \cite[Section~3.1]{Voll10}.
	In accordance with \cite[Section~3]{Voll10}, we finally define 
	\[
	b_{\varepsilon}(X)=\binom{d}{I}_{X^{-1}}X^{\sum_{\iota\in I}r_{\iota}\iota(d-\iota)}.
	\]
	
	\begin{figure}[h]
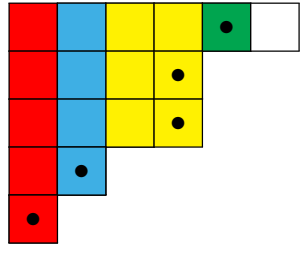

\centering
\ydiagram[*(yellow) \bullet]
  {0,3+1,3+1}
  *[*(white)]{5+1}
  *[*(Green) \bullet ]{4+1}
   *[*(red) \bullet ]{0,0,0,0,1}
    *[*(CornflowerBlue) \bullet ]{0,0,0,1+1}
*[*(CornflowerBlue)]{1+1,1+1,1+1,1+1}
*[*(yellow)]{2+2,2+1,2+1}
*[*(red)]{5,4,3,2,1}
\caption{In this figure $\varepsilon=(5,4,3,3,1,0)$. The different colors represent the different $\tilde{\varepsilon}_j$'s. The horizontal shifts represent the $i_j$'s while the vertical dots represent the $r_{i_j}$'s. Concretely, $\ell=4$ and  $(i_1,i_2,i_3,i_4)=(1,2,4,5)$ and $(r_{i_1},r_{i_2},r_{i_3},r_{i_4})=(1,1,2,1)$.}
        \label{fig:tableau}
\end{figure}	

    \noindent	
	The next result is a direct consequence of the work in \cite[Section~3]{Voll10}; cf.\ in particular \cite[Equation~(26)]{Voll10}.
	
    \begin{proposition}\cite[Section~3]{Voll10}\label{prop:ball-voll}
    Let 
    $[L]\in\mathcal{B}_d^0(K)$. Then the following hold: 
    \begin{enumerate}[label=$(\arabic*)$]
    \item\label{it:ball-voll1} for each $\varepsilon\in\mathcal{E}_r^{(d)}$, one has that $b_{\varepsilon}(X)$ is a monic integral polynomial of degree
    \[\deg b_{\varepsilon}(X)=\sum_{\iota\in I(\varepsilon)}r_\iota\iota(d-\iota)=|\Ocal_K^{d\times d}:\Lambda_{\varepsilon}|.\]
    \item\label{it:ball-voll2} 
    one has $\lvert\B_r([L])\rvert =b_r^{(d)}(q)=\sum_{\varepsilon\in\mathcal{E}_r^{(d)}}b_{\varepsilon}(q)$.
    
    \end{enumerate}
    \end{proposition}

	\begin{example}\label{ex:counting}
	 For $[L]\in\mathcal{B}_3^0(K)$, we have
	 \begin{align*}
 \lvert \B_2([L]) \rvert =& b^{(3)}_2(q) \\ =&  b_{(0,0,0)}(q)+(b_{(1,0,0)}(q)+b_{(1,1,0)}(q)) + (b_{(2,0,0)}(q)+b_{(2,1,0)}(q)+b_{(2,2,0)}(q))\\
	      = & 1+2(q^2+q+1)+(2(q^4+q^3+q^2)+(q^4+2q^3+2q^2+q))\\ 
	      = & 3q^4+4q^3+6q^2+3q+3.
	 \end{align*}
	\end{example}
	
	

\begin{definition}
Let $\rev:\Z^d\rightarrow\Z^d$ be the involution defined by 
\[\varepsilon=(\varepsilon_1,\ldots,\varepsilon_d)\longmapsto\rev(\varepsilon)=(\varepsilon_d,\ldots,\varepsilon_1).\]
\end{definition}

\begin{lemma}\label{lem:ereve}
Let $\lambda$ be a non-negative integer and let $\varepsilon,\varepsilon'\in\mathcal{E}_r^{(d)}$. The following hold:
\begin{enumerate}[label=$(\arabic*)$]
\item\label{it:ereve1} If $\varepsilon+\rev{\varepsilon'}=\lambda{\bf 1}$, then $\deg b_{\varepsilon}(X)=\deg b_{\varepsilon'}(X)$.
\item\label{it:ereve2} If $k\in\{1,\ldots,d\}$ is such that 
\[
\varepsilon-\varepsilon'=(\delta_{ik}\lambda)_{i=1,\ldots,d}
\]
then $\deg b_{\varepsilon}(X)=\deg b_{\varepsilon'}(X)+(d+1-2k)\lambda$.
\end{enumerate}
\end{lemma}

\begin{proof}
(1) Assume that $\varepsilon+\rev{\varepsilon'}=\lambda{\bf 1}$, equivalently, for all $i\in\{1,\ldots,d\}$, one has $\varepsilon_i=\lambda-\varepsilon_{d-i+1}$. It follows from \Cref{prop:ball-voll}\ref{it:ball-voll1} that
\begin{align*}
\deg b_{\varepsilon}(X)  = |\Ocal_K^{d\times d}:\Lambda_{\varepsilon}|&=\sum_{1\leq i<j\leq d}\varepsilon_i-\varepsilon_j
=\sum_{1\leq i<j\leq d}\varepsilon'_{d-j+1}-\varepsilon'_{d-i+1}\\
&=\sum_{1\leq s<t\leq d}\varepsilon'_s-\varepsilon'_t=|\Ocal_K^{d\times d}:\Lambda_{\varepsilon'}|=\deg b_{\varepsilon'}(X).
\end{align*}
(2) Let $k\in\{1,\ldots,k\}$ be such that 
\[
\varepsilon_s=\begin{cases}
\varepsilon_s' & \textup{ if } s\neq k, \\
\varepsilon_s'+\lambda & \textup{ if } s=k.
\end{cases}
\]
It follows from \Cref{prop:ball-voll}\ref{it:ball-voll1} that
\begin{align*}
\deg b_{\varepsilon}(X)  = |\Ocal_K^{d\times d}:\Lambda_{\varepsilon}|&=\sum_{1\leq i<j\leq d}\varepsilon_i-\varepsilon_j\\ 
&=\sum_{\substack{1\leq i<j\leq d\\ i,j\neq k}}\varepsilon'_i-\varepsilon'_j
+\sum_{k<j\leq d}(\varepsilon_k'+\lambda-\varepsilon_j')+\sum_{1\leq i<k}(\varepsilon_i'-\varepsilon_j'-\lambda)
\\
&=\sum_{1\leq s<t\leq d}\varepsilon'_s-\varepsilon'_t+(d+1-2k)\lambda\\
 &=|\Ocal_K^{d\times d}:\Lambda_{\varepsilon'}|+(d+1-2k)\lambda=\deg b_{\varepsilon'}(X)+(d+1-2k)\lambda.
\end{align*}
\end{proof}

\noindent
The proof of the next result shows that the asymptotics of $|\B_r([L])|$ is dominated by $|\partial\B_r([L])|$, i.e.\ the dominating summands in $b_r^{(d)}(X)$ correspond to  elements of $\partial\mathcal{E}_r^{(d)}$. 

\begin{theorem}\label{th:asymptotics}
The following hold:
\begin{enumerate}[label=$(\arabic*)$]
\item If $d$ is even, then the leading term of $b_r^{(d)}(X)$ is $X^{d^2r/4}$.
\item If $d$ is odd, then the leading term of $b_r^{(d)}(X)$ is $(r+1)X^{(d^2-1)r/4}$. 
\end{enumerate}
\end{theorem}

\begin{proof}
We prove (2). To this end, write $d=2k+1$ and define the subset 
$\mathcal{S}$ of $\mathcal{E}_r^{(d)}$ to consist of all elements $\varepsilon$ satisfying
\[
\varepsilon_i=\begin{cases}
r & \textup{ if } i< k+1,\\
0 & \textup{ if } i> k+1.
\end{cases}
\] 
Then $\mathcal{S}$ has cardinality $r+1$. Let moreover, $\mathcal{S}_{-}$ and $\mathcal{S}_+$ denote the subsets of $\mathcal{E}_r^{(d)}$ of those elements that are smaller resp.\ bigger than elements in $\mathcal{S}$, with respect to the lexicographic order. Then $\mathcal{E}_r^{(d)}$ equals the disjoint union $\mathcal{S}_-\cup\mathcal{S}\cup\mathcal{S}_{+}$. Let now $\varepsilon\in\mathcal{E}_r^{(d)}$ and $\varepsilon^*\in\mathcal{S}$. If $\varepsilon\in\mathcal{S}_-$, then \Cref{lem:ereve}\ref{it:ereve2} yields that $\deg b_{\varepsilon}(X)<\deg b_{\varepsilon^*}(X)$. Moreover, \Cref{lem:ereve}\ref{it:ereve2} also ensures that, if $\varepsilon\in\mathcal{S}$, then $\deg b_{\varepsilon}(X)=\deg b_{\varepsilon^*}(X)$. Assume now that $\varepsilon\in\mathcal{S}_+$: we claim that $\deg b_{\varepsilon}(X)>\deg b_{\varepsilon^*}(X)$. To this end, define $\varepsilon'=\rev(r{\bf 1}-\varepsilon)$ and note that $\varepsilon'\in \mathcal{S}_-$. 
Now $\deg b_{\varepsilon}(X)>\deg b_{\varepsilon^*}(X)$ thanks to \cref{lem:ereve}\ref{it:ereve1} and so we conclude thanks to \Cref{prop:ball-voll}\ref{it:ball-voll1}.

To prove (1), one can proceed in an analogous way by defining $\mathcal{S}$ to be the singleton consisting of the vector whose first $d/2$ entries are equal to $r$ and all others are $0$.
\end{proof}

\begin{remark}[Asymptotic of balls against Sperner codes]\label{rmk:asymptotics}
We have seen in \cref{sec:sperner-codes} that, if $\mathcal{C}$ is a Sperner code with parameters $(d;R;\alpha)$ and $e=\lceil d/2\rceil$, then the cardinality of $\Ccal$ is the same as that of $\Gr(e,V_{r+\alpha-1})$. In particular, thanks to \cref{prop:ball-voll}(1), we know that the leading term of the polynomial describing $|\Ccal|$ is equal to $q^{(r+1-\alpha)e(d-e)}$. Rewriting thus compactly the degree of the leading terms from \cref{th:asymptotics} as $re(d-e)$, we get that the density of a Sperner code on $\partial\B_r$ is asymptotically equivalent (as $q\to\infty$) to
\[
q^{(1-\alpha)e(d-e)}\cdot\begin{cases}
1 & \textup{ if } d \textup{ is even}, \\
(r+1)^{-1} & \textup{ otherwise}.
\end{cases}
\]
\end{remark}

\begin{remark}[Analogue of sphere packing bounds for odd distances]\label{rmk:sphere-pack-bounds}
Let $\Ccal$ be a spherical code in $\B_r$, as defined in \cref{def:spherical-code-building}, of odd minimum distance $2\alpha+1$. In this case, it is clear that any two elements $[L]$ and $[L']$ of $\Ccal$ satisfy $\B_{\alpha}([L])\cap\B_\alpha([L'])=\emptyset$. It follows therefore that a very loose sphere packing bound on the cardinality of $\Ccal$ is given by
\[
|\Ccal|\leq \frac{|\B_{r+\alpha}|-|\B_{r-\alpha}|}{|\B_{\alpha}|-1},
\]
which indeed, thanks to \cref{th:asymptotics}, is asymptotically no better that the known trivial bound given by $|\partial\B_r|$. For a better asymptotic bound  one should compute, for $[L]\in\partial\B_r$ the size of $\B_\alpha([L])\cap \partial\B_r$ yielding the tighter
\[
|\Ccal|\leq \frac{|\partial\B_r|}{|\B_\alpha([L])\cap \partial\B_r|};
\]
compare with \cite[Theorem~1.6.1]{EriZin/01}.
What is the asymptotic behaviour of the right term of the last inequality as $q\to\infty$?
\end{remark}

    \bigskip \noindent
	{\bf Acknowledgements}.
I wish to thank Gabriele Nebe for suggesting the investigation of spherical codes in the context of buildings and for her precious feedback on the content of this paper. In addition, I am thankful to Christopher Voll for pointing out and illustrating the results from \cite{Voll10}. I am grateful to Yassine El Maazouz, Alessandro Neri, Bernd Sturmfels, and Christopher Voll for their very helpful comments on an early version of this manuscript.  I thank the two anonymous referees for their valuable reports, which helped improve the exposition of this paper. This project was supported by the Deutsche Forschungsgemeinschaft 
	(DFG, German Research Foundation) -- Project-ID 286237555 -- TRR 195.

\end{document}